\documentclass[10pt,twocolumn]{IEEEtran}

\usepackage{graphics}
\usepackage{graphicx}
\usepackage{grffile}
\usepackage{amsmath}
\usepackage{amsthm}
\usepackage{amssymb}

\newtheorem{thm}{Theorem}

\newtheorem{lem}{Lemma}
\newtheorem{cla}{Claim}
\newtheorem{mydef}{Definition}

\newtheorem{corr}{Corollary}
\newtheorem{prop}{Proposition}
\newtheorem{rem}{Remark}

\title{Decoding Complexity of Irregular LDGM-LDPC Codes Over the 
BISOM Channels}

\author{Manik Raina and Predrag Spasojevi\'{c} \\
WINLAB, Rutgers University, North Brunswick, NJ, 08901.\\
\texttt{manikraina@gmail.com, spasojev@winlab.rutgers.edu}
}

\begin{document}
\maketitle

\begin{abstract}
An irregular LDGM-LDPC code is studied as a sub-code of an LDPC code with some randomly \emph{punctured} 
output-bits. It is shown that the LDGM-LDPC codes achieve rates arbitrarily 
close to the channel-capacity of the binary-input symmetric-output
 memoryless (BISOM) channel  with bounded \emph{complexity}. The measure of
 complexity is the average-degree (per information-bit) of the check-nodes for the
 factor-graph of the code. A lower-bound on the average degree of the 
check-nodes of the irregular LDGM-LDPC codes is obtained. The bound does not depend on the decoder used at the receiver.  The stability condition for decoding the irregular LDGM-LDPC codes over the binary-erasure channel (BEC) under iterative-decoding with message-passing is described.  
\end{abstract}
\begin{IEEEkeywords}
Bounded-complexity-codes, iterative-decoding, capacity-approaching-codes.
\end{IEEEkeywords}

\let\thefootnote\relax\footnotetext{This research was supported by the National Science Foundation under Grant CNS-07-21888.} 
\section{Introduction}

Two questions guide much of the research in channel-coding: the construction
 of codes that achieve rates arbitrarily close to the capacity of a given
 channel and efficient decoding of these codes. The decoding of 
error-correcting-codes using message-passing over sparse-graphs is
 considered the state-of-the-art. An example of a sparse-graph code is the
 ensemble of low-density parity-check codes \cite{gall}. Consider a 
binary-input symmetric-output memoryless (BISOM) channel with 
channel-capacity $C$. Suppose a code is chosen at random from a given 
code-ensemble and achieves a rate 
$(1-\epsilon)C$, where $\epsilon \in (0,1]$ is the multiplicative
 gap-to-capacity. The study of the encoding and decoding complexity
 of code-ensembles in terms of the capacity-gap $\epsilon$ was proposed by
 Khandekar and McElice \cite{khan}. \par  Low-density parity-check (LDPC)
 codes exhibit remarkable performance under message-passing decoding. This
 performance is attributed to the sparseness of the 
parity-check matrices of these codes. The \emph{density} of a parity-check
 matrix is the number of ones in the parity-check matrix 
per-information-bit. The density is proportional to the number of
 messages passed in one round of iterative-decoding.  A
 lower-bound on the density of a parity-check matrix in terms of the
 multiplicative-capacity-gap $\epsilon$ was obtained in \cite{sason1} and later tightened in \cite{sason}. For a code defined by a 
full-rank parity-check matrix, the lower-bound on the density is $\frac{K_1 + K_2 \log \frac{1}{\epsilon}}{1-\epsilon}$, where $K_1$
 and $K_2$ depend on the channel and not on code parameters. As the rate of code approaches the channel-capacity ($\epsilon \rightarrow 0$),
 the density of the parity-check matrix becomes unbounded. The authors of
 \cite{phister} showed that non-systematic irregular-repeat-accumulate
 (NSIRA) codes could achieve rates arbitrarily close to the channel-capacity
 of a BISOM channel with bounded complexity. The rates close to channel-capacity were achieved by randomly puncturing the information bits of the NSIRA codes indepedently with a probability that depended on the gap to capacity. Recently, the authors of 
\cite{liupar}, \cite{parallel} modeled several communication scenarios using parallel
 channels. This model enables (among other things) the investigation of the
 performance of punctured LDPC codes. The effect of random-puncturing on the
 ensemble of  $(j,k)$ regular LDPC codes was studied in \cite{parallel-hsu};
 an upper-bound on the weight spectrum of the ensemble of LDPC codes in
 question was obtained.  The ensemble of 
low-density generator-matrix/low-density parity-check (LDGM-LDPC) codes was
 studied in \cite{hsu}, \cite{wain}. This ensemble resultes on compounding the LDGM and  LDPC codes. Hsu \cite{hsu}
 proved that codes from the {\it regular} LDGM-LDPC ensemble could achieve rates
 arbitrarily close to the channel-capacity of the BISOM channel with bounded
 graphical complexity. However, the proof of \cite{hsu} assumed: a
 regular LDGM code with rate 1; a regular LDPC code; and, a maximum-likelihood (ML) decoder. No puncturing was employed.
 Pfister and Sason \cite{ara-codes} studied capacity achieving degree-distributions for the accumulate-repeat-accumulate (ARA) codes over the BISOM
 channel. Using a technique called \emph{graph-reduction}, some 
capacity-achieving degree-distributions for accumulate LDPC (ALDPC) codes
 were proposed. ALDPC codes were shown to be LDGM-LDPC codes with a 
2-regular LDGM code. In this work, the upper LDGM code can have any rate
 $R_G \in (0,1]$. Further, the LDPC and LDGM codes can be irregular and the
 requirement for ML decoding is removed.\par This paper obtains lower-bounds on
the complexity of the ensemble of irregular LDGM-LDPC codes at rates arbitrarily
close to the capacity of the binary-input symmetric-output memoryless channel for asymptotic block-lengths. The information-theoretic bounds obtained in this paper do not depend on the type of decoder. The LDGM-LDPC codes are studied as sub-codes of \emph{constrained punctured LDPC codes}. It is shown that if some variable nodes of the constrained punctured LDPC codes are punctured independently with probability $p = 1 - \kappa \epsilon$ (for some constant $\kappa$), the ensemble achieves rates arbitrarily close to channel capacity of BISOM channel with \emph{bounded complexity}. Further, it is shown that that the LDGM-LDPC codes are equivalent to the constrained punctured LDPC codes when $\epsilon \rightarrow 0$ (or $p \rightarrow 1$). The performance of the constrained punctured LDPC codes are studied over the binary-erasure channel under iterative-decoding using message-passing. The stability conditions are derived.\par This paper is organized as follows. Some preliminary topics are introduced in Section \ref{prelims}. LDGM-LDPC codes are modeled as sub-codes of constrained punctured LDPC codes in Section \ref{special}. Performance of the  constrained punctured LDPC codes and bounds on the average degrees of the factor graph are studied in Section \ref{punc-sec}. The stability condition
for these codes over the binary-erasure channel under message-passing decoding is studied in Section \ref{sta-bec}. The paper is concluded in Section \ref{conclusion}.

\section{Preliminaries}
\label{prelims}
In this paper, uppercase, lowercase and bold-uppercase variables represent
 random-variables, realization of random variables and random-vectors 
 respectively. For example, $X$ is a random-variable with a realization $x$
 while $\bf{X}$ is a random-vector. 

        \subsection{LDGM-LDPC Codes}
        \label{ldgm-ldpc-codes}
Regular LDGM-LDPC codes were studied in \cite{hsu,wain}. In this paper,
 irregular LDGM-LDPC codes are studied. Consider the binary random vectors
 ${\bf X_1,X_2}$ of length $n^{[1]}$ and $n^{[2]}$ respectively. The 
LDGM-LDPC code is defined as follows:
\begin{equation}
\label{ldgm-ldpc}
        \mathcal{C} \overset{\Delta}{=} \{ {\bf X_1} :  {\bf X_1} = {\bf X_2}G  ,  {\bf X_2} H^T = {\bf 0} \}  
\end{equation}
where $H$ and $G$ are the random low-density parity-check (LDPC) matrix and
 random low-density generator-matrix (LDGM) respectively. Consider the
 factor-graphs $\cal{G}_H$ and $\cal{G}_G$ represented by the matrices $H$
 and $G$ respectively. 
\begin{figure}[htp]
\centering
\includegraphics[scale=0.3]{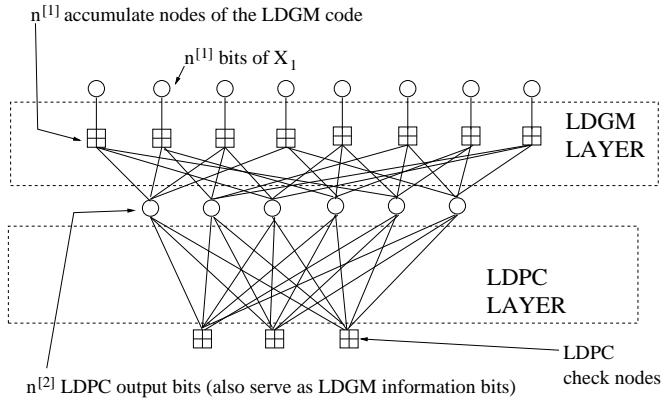}
\caption{The LDGM-LDPC Code}\label{facto}
\end{figure}
Let $\cal{G}_H$ be a $(n^{[2]}, \lambda_H(x),\rho_H(x))$ factor-graph with 
$n^{[2]}$ variable-nodes where we define the following generating-functions:
\begin{align}
\label{lambdah}
\lambda_H(x) &= \sum_{i} \lambda_{H,i} x^{i-1} , &   \rho_H(x) &= \sum_{i} \rho_{H,i} x^{i-1}
\end{align}
where $\lambda_{H,i}$ ($\rho_{H,i}$) is the probability of a randomly chosen
 edge in $\cal{G}_H$ being connected to a variable (check) node of degree 
$i$. Similarly, Let 
$\cal{G}_G$ be a $(n^{[1]}, \lambda_G(x),\rho_G(x))$ factor-graph with 
$n^{[1]}$ accumulated-nodes where we define the following 
generating-functions:
\begin{align}
\label{lambdag}
\lambda_G(x) &= \sum_{i} \lambda_{G,i} x^{i-1} , &   \rho_G(x) &= \sum_{i} \rho_{G,i} x^{i-1}
\end{align}
where $\lambda_{G,i}$ ($\rho_{G,i}$) is the probability of a randomly chosen
 edge in $\cal{G}_G$ being connected to a information (accumulate) node of
 degree $i$. The two factor graphs $\cal{G}_G$ and $\cal{G}_H$ are compounded to form the LDGM-LDPC code as shown in figure \ref{facto}.

\section{LDGM-LDPC Codes As a Special LDPC Code}
\label{special}
In this section, the LDGM-LDPC code as defined in (\ref{ldgm-ldpc}) is
 shown to be a sub-code of an LDPC code.

\begin{lem}
\label{ldgmlem}
Consider the binary vector ${\bf X} = ({\bf X_1}{\bf X_2})$ that results on the
 concatenation of the two binary vectors ${\bf X_1}$ and ${\bf X_2}$ (of lengths $n^{[1]}$ and $n^{[2]}$ respectively),  which satisfy \eqref{ldgm-ldpc}. A new constraint-matrix is defined as follows:
\begin{align}
\label{hmatrix}
\mathcal{H} = \begin{pmatrix} I\ \ \ 0 \\  \\\ G \ \ \ H^T \end{pmatrix} ^T
\end{align}
where $I$ is the $n^{[1]} \times n^{[1]}$  identity matrix, $G$ and $H$ are
 the random LDGM generator-matrix and LDPC parity-check matrix of the LDGM-LDPC code (defined in (\ref{ldgm-ldpc})). Then, the code 
$\mathcal{C}_{\mathcal{H}} \overset{\Delta}{=} \{ {\bf X} : {\bf X} \mathcal{H}^T = {\bf 0} \}$ is a parity-check code with
 parity-check matrix $\mathcal{H}$ which is a \emph{mother code} of the LDGM-LDPC code defined in \eqref{ldgm-ldpc}.
\end{lem}
\begin{proof}
It follows from (\ref{hmatrix}) that 
\begin{align*}
\{ ({\bf X_1}{\bf X_2}) \mathcal{H}^T ={\bf 0} \}  \iff \{  {\bf X_1} = {\bf X_2}G  \text{ and }  {\bf X_2} H^T = {\bf 0} \}
\end{align*}
where all arithmetic is over 
${\bf \text{GF}}(2)$.  
The bits ${\bf X_1}$ are identical to the code bits of the LDGM-LDPC code in \eqref{ldgm-ldpc}. Thus, the code in \eqref{ldgm-ldpc} is a sub-code of the code $\mathcal{C_{\mathcal{H}}}$. The vector ${\bf X}$ is a length $n$ parity-check code with a sparse parity-check matrix $\mathcal{H}$.  $\mathcal{H}$ is the parity-check matrix of a randomly chosen code from the ensemble $(n, \lambda_G(x),\rho_G(x),\lambda_H(x),\rho_H(x))$. 
\end{proof}

\section{Puncturing the Code $\mathcal{C}_{\mathcal{H}}$}
\label{punc-sec}
In this section, a random puncturing scheme is introduced for the code 
$\mathcal{C}_{\mathcal{H}}$ that was defined in lemma \ref{ldgmlem}. Further, the lower-bound on the average density of the irregular LDGM-LDPC ensemble is obtained. Consider a length $n$ codeword ${\bf X} = \{X_1,\hdots,X_n\}$ that is transmitted 
over a BISOM channel. A code bit of ${\bf X}$ is punctured if the output 
at the BISOM channel corresponding to the said code bit is $0$. Some puncturing schemes for codes were proposed in \cite{ha}, which include random puncturing (codeword bits were punctured independently with some probability $p$) or intentional-puncturing (code bits were divided into classes and each class had its own puncturing probability).
\begin{rem}
The codeword ${\bf X}$ is assumed to be uniformly chosen from the code 
$\mathcal{C}_{\mathcal{H}}$. It is assumed
in this paper that all the bits of the codeword are equally likely to be $0$
 or $1$. 
\end{rem}
The result \cite[Proposition 2.1]{parallel} is now restated. It is assumed that every codeword bit in ${\bf X}$ is transmitted though one of the $J$ statistically independent BISOM channels, where $C_j$ is the capacity of the $j$th channel (in bits per channel use) and $p_{Y|X}(.|.;j)$ is the transition probability of the $j$th channel. Let the received message at the channel output be ${\bf  Y}$. The conditional probability-density of the log-likelihood ratio $\log \frac{p_{Y|X}(Y=y|0;j)}{p_{Y|X}(Y=y|1;j)}$ at the output of the $j$th channel given the input is 0 is denoted by $a(.;j)$. Let $\mathcal{I}(j)$ be the set of indices of the code bits  transmitted over the $j$th channel, 
$n^{[j]} \overset{\Delta}{=} |\mathcal{I}(j)|$ be the size of this set, and 
$p_j = \frac{n^{[j]}}{n}$ be the fraction of bits transmitted over the $j$th 
channel. For an arbitrary $c \times n$ parity-check matrix $H$ of the code 
$\mathcal{C}$, let $\beta_{j,m}$ designate the number of indices in 
$\mathcal{I}(j)$ referring to bits which are involved in the $m$th 
parity-check equation of $H$ and let $R_d = 1 - \frac{c}{n}$ be the design
rate of $\mathcal{C}$.

\begin{prop}
\label{propsason}
Let $\mathcal{C}$ be a binary linear block code of length $n$, and assume 
that its transmission takes place over a set of $J$ statistically independent
BISOM channels. Let ${\bf X} = \{X_1,\hdots, X_n\} \text{ and } {\bf Y} = \{Y_1,\hdots, Y_n\}$ designate the transmitted codeword and received sequence
respectively.  Then, the average conditional entropy of the transmitted 
codeword given the received sequence satisfies
\begin{align*}
\frac{1}{n}H({\bf X}|{\bf Y}) \ge 1 - \sum_{j = 1}^J p_j C_j - (1 - R_d) \\
. \Biggl(1 - \frac{1}{2n(1 - R_d)\log 2} \\
\sum_{p = 1}^{\infty} \left\{ \frac{1}{p(2p -1)} \sum_{m = 1}^{n(1 - R_d)} 
\prod_{j = 1}^J (g_{j,p})^{\beta_{j,m}} \right\} \Biggr) 
\end{align*} 
where
\begin{equation}
\begin{split}
\label{defsdefs}
g_{j,p}  & \overset{\Delta}{=} \int_0^{\infty} a(l;j)(1+e^{-l})\text{tanh}^{2p} \Bigl( \frac{l}{2} \Bigr) dl, \\& j \in \{1,\hdots,J\}, p \in \mathbb{N}.
\end{split}
\end{equation}
\end{prop}
\begin{mydef} \label{cph} \emph{Constrained punctured LDPC code} $\mathcal{C}_{\mathcal{H}}(p)$: Let $\mathcal{C}_{\mathcal{H}}$ be a parity-check code defined in lemma \ref{ldgmlem}. If the first $n^{[1]}$ bits of this code (${\bf X_1}$) pass through the channel without puncturing and the last $n^{[2]}$ bits of the code (${\bf X_2}$) are punctured independently with probability $p$, the resulting code is represented by $\mathcal{C}_{\mathcal{H}}(p)$.    
\end{mydef}
The following remark explains why the above punctured LDPC codes are termed "constrained".
\begin{rem}
Let $k_m$ be a random variable representing the number of edges involved in the $m$th parity-check of a given parity-check code. In the bound derived in proposition \ref{propsason}, $\beta_{j,m}$ refers to the number of code bits from the $j$th class that are connected to the $m$th parity-check. For every $m$, it follows that:
\begin{align*}
k_m = \underset{j \in [1,\hdots,J]}{\sum} \beta_{j,m}
\end{align*}
where $J$ is the number of parallel, statistically independent channels. When discussing the code $\mathcal{C}_{\mathcal{H}}(p)$, $J = 2$. The case $j = 1$ corresponds to the $n^{[1]}$ un-punctured bits (variable nodes) ${\bf X_1}$ and $j = 2$ corresponds to the $n^{[2]}$ bits (variable nodes) of ${\bf X_2}$ that are independently punctured with probability $p$. From the structure of the code $\mathcal{C}_{\mathcal{H}}(p)$, each of the first $n^{[1]}$ parity-checks are connected to exactly one variable node from the first class (bits of ${\bf X_1}$), i.e $\beta_{1,m} = 1$, if $m \in [1,n^{[1]}]$. The number of variable nodes connected to the first $n^{[1]}$ check nodes from $j=2$ is $\beta_{2,m} = k_m - 1$, where $k_m$ is distributed as per $\rho_G(.)$ of \eqref{lambdag} (see figure \ref{facto} and \ref{factor-g-pic}). The remaining parity-checks of the code  $\mathcal{C}_{\mathcal{H}}(p)$ are connected to variables nodes from the second class (bits of ${\bf X_2}$) only. Thus if $m \ge n^{[1]}$, $\beta_{1,m} = 0$ and $\beta_{2,m} = k_m$, where $k_m$ is distributed as per $\rho_H(.)$ of \eqref{lambdah}. To summarize:
\begin{equation}
\begin{split}
\label{betas}
        &\beta_{1,m} = \begin{cases} 1 ,   m \in [1, n^{[1]}], \\ 0 ,  \text{ otherwise} \end{cases}  \\
        & \beta_{2,m} = \begin{cases} k_m -1 , m \in [1, n^{[1]}] \\ k_m , \text{ otherwise} \end{cases}
\end{split}
\end{equation}
In parallel LDPC codes of \cite{parallel}, the members of the sequence $\{\beta_{1,m},\hdots, \beta_{J,m}\}$ take on all possible values between $1$ and $k_m$ such that $\sum_j \beta_{j,m} = k_m$. On the other hand, for the code $\mathcal{C}_{\mathcal{H}}(p)$, $\beta_{1,m}$ takes values $0$ or $1$ only. This follows from: the structure of the parity-check matrix $\mathcal{H}$ of the code $\mathcal{C}_{\mathcal{H}}$ (and of $\mathcal{C}_{\mathcal{H}}(p)$); and, the assignment of ${\bf X_1}$ and ${\bf X_2}$ to the two classes of parallel channels. 
\end{rem}
\begin{cla}
Let $\mathcal{C}$ and $\mathcal{C}_{\mathcal{H}}(p)$ represent the codes defined in \eqref{ldgm-ldpc} and definition \ref{cph} respectively. Then,
\begin{align*}
\underset{p \rightarrow 1}{\text{lim}} \mathcal{C}_{\mathcal{H}}(p) = \mathcal{C}
\end{align*}
\end{cla}
The above claim follows from \eqref{ldgm-ldpc} and definition \ref{cph} because in the limit $p \rightarrow 1$, all the bits of the lower LDPC code ${\bf X_2}$ are punctured. The codewords of $\mathcal{C}_{\mathcal{H}}(p)$ are ${\bf X_1}$, which is identical to the codewords of $\mathcal{C}$ of \eqref{ldgm-ldpc}.\par The following lemma relates the code rate and the conditional entropy 
$H({\bf X}|{\bf Y})$ of the code defined in definition \ref{cph}.
\begin{lem}
\label{pblem}
Let $\mathcal{C_{\mathcal{H}}}(p)$ be a code of length $n$  as defined in 
definition \ref{cph}. Let ${\bf X}$ be a binary codeword from  
$\mathcal{C_{\mathcal{H}}}(p)$. Let ${\bf Y}$ be a vector sequence of length 
$n$ at the output of the BISOM channel upon transmission of ${\bf X}$. Then, 
the following inequality holds:
\begin{align}
\label{lemmaplaceholder}
\frac{1}{n}H({\bf X}|{\bf Y}) \le \frac{n^{[2]}}{n}H(P_b)  
\end{align}  
where $P_b$ is the average bit-error probability of decoding the lower LDPC code ${\bf X_2}$  
, $n^{[2]}$ is the length of the lower LDPC code ${\bf X_2}$, as defined in 
(\ref{ldgm-ldpc}). 
\end{lem}
\begin{proof}
\begin{equation}
\begin{split}
\label{useless1}
        &\frac{1}{n}H({\bf X}|{\bf Y}) \overset{a}{=} \frac{1}{n}H({\bf X_1},{\bf X_2}|{\bf Y}) 
\\& \overset{b}{=}\frac{1}{n}H({\bf X_2}|{\bf Y})+ \frac{1}{n}H({\bf X_1}|{\bf Y},{\bf X_2}) \overset{c}{=}\frac{1}{n}H({\bf X_2}|{\bf Y})
\end{split}
\end{equation}
where $\overset{a}{=}$ follows from the definition of ${\bf X}$, $\overset{b}{=}$ 
follows from the chain rule of entropy and $\overset{c}{=}$ follows because
for a given code $\mathcal{C}_{\mathcal{H}}(p)$, the entropy of $\bf{X_1}$ is zero if $\bf{X_2}$ is known (this follows from ${\bf X_1} = {\bf X_2} G$). Further,
\begin{equation}
\begin{split}
\label{useless2}
\frac{1}{n}H({\bf X_2}|{\bf Y}) &\overset{d}{\le} \frac{1}{n}\sum_{i = 1}^{n^{[2]}} h_2(p^i_e) = \frac{n^{[2]}}{n}\frac{1}{n^{[2]}}\sum_{i = 1}^{n^{[2]}} h_2(p^i_e) \\ & \overset{e}{\le} \frac{n^{[2]}}{n} h_2\Bigl(\frac{1}{n^{[2]}}\sum_{i = 1}^{n^{[2]}}p^i_e\Bigr) \overset{f}{=} \frac{n^{[2]}}{n} h_2(P_b) 
\end{split}
\end{equation}
where $\overset{d}{\le}$ follows from the Fano's inequality for binary valued
random variables and where $p^i_e$ is the bit error probability for the $i$th bit of ${\bf X_2}$, $\overset{e}{\le}$ follows from the concavity of the binary
entropy function and $\overset{f}{=}$ follows from the definition of the 
average bit-error probability of ${\bf X_2}$. \eqref{lemmaplaceholder} follows from (\ref{useless1}) and (\ref{useless2}).
\end{proof}
In the following theorem, it is assumed that the code length $n \rightarrow \infty$ and $P_b \rightarrow 0$. An upper-bound on the design-rate for the ensemble of parity-check codes defined in lemma \ref{ldgmlem} is obtained.   
\begin{thm}
\label{ratebound}
Consider a $(n, \lambda_G(x),\rho_G(x),\lambda_H(x),\rho_H(x))$ ensemble as defined in lemma \ref{ldgmlem}. Let ${\bf X}$ be a codeword from the code $\mathcal{C_{\mathcal{H}}}(p)$ that is chosen uniformly from this ensemble. Let the first $n^{[1]}$ bits of ${\bf X}$ pass through a BISOM channel with capacity $C$ without puncturing. The last $n^{[2]}$ bits of ${\bf X}$ pass through the BISOM channel after being punctured independently with probability $p$. Let $p_1 = \frac{n^{[1]}}{n}$, $p_2 = \frac{n^{[2]}}{n}$ (where $p_1 + p_2 = 1$) and let $R_H$ be the design-rate of the lower LDPC code in the LDGM-LDPC code. Further, let $a_L$ and $a_R$ be the average degrees of the accumulate nodes of the LDGM codes and check nodes of the LDPC nodes respectively. Then, the design rate $R_d$ of the ensemble is upper-bounded as:
\begin{align*}
R_d \le 1 - \frac{1 - (p_1 + (1-p)p_2)C}{1 - \frac{1}{2 \log 2} \frac{g_{1,1}p_1 + (1-R_H)p_2 }{p_1 + (1-R_H)p_2} g_{2,1}^{\frac{g_{1,1}p_1 a_L + (1-R_H)p_2 a_R}{g_{1,1}p_1 + (1-R_H)p_2}}}
\end{align*} 
where $g_{1,1}$ and $g_{2,1}$ are defined as per (\ref{defsdefs}).
\end{thm}
The above theorem is proved in the appendix. The above upper-bound on the design-rate of punctured $(n, \lambda_G(x),\rho_G(x),\lambda_H(x),\rho_H(x))$ ensembles (as defined in lemma \ref{ldgmlem}) can be used to obtain a lower-bound on the asymptotic complexity of the code. In the following theorem, the lower-bound is obtained.

\begin{corr}
\label{complexitybound}
In the limiting case of $n \rightarrow \infty$, puncturing the last $n^{[2]}$ bits of a codeword (independently with probability $p$) results in a channel capacity $\overset{-}{C} = (1 - p_2 p)C$, where $C$ is the capacity of the BISOM channel under consideration. Let $a_L$ and $a_R$ be the average degrees of the LDGM accumulate nodes and the LDPC check nodes. Let $R_H$ be the rate of the lower LDPC code. Then, if the design rate of the ensemble $R_d = (1-\epsilon)\overset{-}{C}$, the following lower-bound on $a_L$ and $a_R$ holds:
\begin{align*}
\frac{p_1  g_{1,1} a_L + (1-R_H)p_2 a_R}{p_1  g_{1,1}+ (1-R_H)p_2}  &\ge \\ \frac{\log\Bigl(\frac{1}{2 \log 2} \frac{p_1  g_{1,1} + (1-R_H)p_2}{p_1+(1-R_H)p_2}\frac{1 - (1-\epsilon)\overset{-}{C}}{\epsilon \overset{-}{C}}\Bigr)}{\log(\frac{1}{g_{2,1}})}
\end{align*} 
\end{corr}
\begin{proof}
The design-rate $R_d$ is set to $(1 - \epsilon)\overset{-}{C}$ in the upper-bound of theorem \ref{ratebound} and obtain the above bound. 
\end{proof}
A direct consequence of the above result is that rates arbitrarily close to channel capacity are possible with finite complexity. 
\begin{lem}
\label{pepsilonlem}
Consider a code \ensuremath{\mathcal{C}_{\mathcal{H}}(p)} discussed in corollary \ref{complexitybound}. Then, in the limit $\epsilon \rightarrow 0$, the lower-bound on the average-degrees is finite if the puncturing probability $p = 1 - \kappa \epsilon$, for some constant $\kappa$. 
\end{lem}
\begin{proof}
Since the last $n^{[2]}$ bits of the code \ensuremath{\mathcal{C}_{\mathcal{H}}(p)}  are punctured independently with probability p, the probability density of the log-likelihood ratio (LLR) of those bits is:
\begin{align}
\label{llrpdf}
        a(l;2) = p \delta_0(l) + (1-p) a(l)
\end{align}
where $\delta_0(l)$ is the Dirac delta function at $l = 0$ and $a(l)$ is the density of the LLR over the original (un-punctured) BISOM channel.
First, $g_{2,1}$ is simplified. As per \eqref{defsdefs}, 
\begin{equation*}
\begin{split}
g_{2,1}   & \overset{\Delta}{=} \int_0^{\infty} a(l;2)(1+e^{-l})\text{tanh}^{2} \Bigl( \frac{l}{2} \Bigr) dl \\& \overset{a}{=} \int_0^{\infty} [p \delta_0(l) + (1-p) a(l)](1+e^{-l})\text{tanh}^{2} \Bigl( \frac{l}{2} \Bigr) dl \\&= (1-p)g_1
\end{split}
\end{equation*}
where $\overset{a}{=}$ follows from \eqref{llrpdf} and where $g_1 = \int_0^{\infty} a(l)(1+e^{-l})\text{tanh}^{2} \Bigl( \frac{l}{2} \Bigr) dl$. In the limit $\epsilon \rightarrow 0$, the lower-bound on the complexity in corollary \ref{complexitybound} is finite if and only if $(1-p)g_1 = \eta \epsilon$ forsome constant $\eta$. Thus, it follows that $p = 1 - \kappa \epsilon$, where $\kappa = \frac{\eta}{g_1}$.
\end{proof}
\begin{lem}
The lower-bound on complexity of the LDGM-LDPC code defined in \eqref{ldgm-ldpc} is finite for a BISOM channel.
\end{lem}
\begin{proof}
Let $\bf{X = (X_1 X_2)}$ represent a randomly chosen codeword from the code $\mathcal{C}_{\mathcal{H}}(p)$. It follows from lemma \ref{pepsilonlem} that if the design rate of this ensemble $R_d$ approaches capacity ($\epsilon \rightarrow 0$) and the puncturing probability $p$ of the LDPC code bits ${\bf X_2}$ approaches $1$, the average lower-bound on the complexity is bounded. In the limit $n \rightarrow \infty$ and $p \rightarrow 1$, the code words of the code $\mathcal{C}_{\mathcal{H}}(p)$ are ${\bf X_1}$ as all the bits of ${\bf X_2}$ are punctured. Thus, the code words of the code $\mathcal{C}_{\mathcal{H}}(p)$ are identical to the LDGM-LDPC code defined in \eqref{ldgm-ldpc}.      
\end{proof}
\section{Stability Condition for Message Passing Decoding of Punctured LDGM-LDPC Codes Over the BEC}
\label{sta-bec}
In this section, the decoding of the ensemble of $(n, \lambda_G(x),\rho_G(x),\lambda_H(x),\rho_H(x))$ codes is studied. The channel is assumed to be a BEC with an erasure probability of $\delta$. It is assumed that the decoder employs iterative-decoding using message-passing. The density-evolution technique of \cite{de} is employed in this work. The main assumption in density-evolution is that the message on an edge of the factor-graph of a randomly chosen code is independent of the messages on all other edges. This assumption is justified because in the asymptotic case $n\rightarrow \infty$, the fraction of bits involved in finite-length cycles vanishes. The density-evolution (DE) equations are obtained for the $l$th stage of decoding. The fixed-point analysis is performed on the DE equations and the stability-condition for DE is derived. \par Consider the factor-graph of the LDGM-LDPC code in fig. \ref{factor-g-pic}. The $l$th iteration of DE is considered. Let $x_1^l$ ($y_1^l$) be the erasure probability along a random edge from (to) the $n^{[1]}$ un-punctured LDGM channel bit nodes to (from) the LDGM accumulate nodes in the $l$th iteration of message-passing. Further, let $x_2^l$ ($y_2^l$) be the erasure probability along a random edge from (to) the $n^{[2]}$ punctured LDPC variable bit nodes to (from) the LDGM accumulate nodes. Similarly, let $x_3^l$ ($y_3^l$) be the erasure probability along a random edge from (to) the $n^{[2]}$ punctured LDPC variable bit nodes to (from) the LDPC check nodes. Consider the $n^{[1]}$ LDGM variable nodes. The messages from the LDGM variable nodes to the LDGM accumulate constraints is an erasure if the original channel symbol that was received was an erasure and the message from the accumulate constraint to the LDGM node in the $l-1$th erasure was an erasure. This observation is formalized as:
\begin{figure}[htp]
\centering
\includegraphics[scale=0.3]{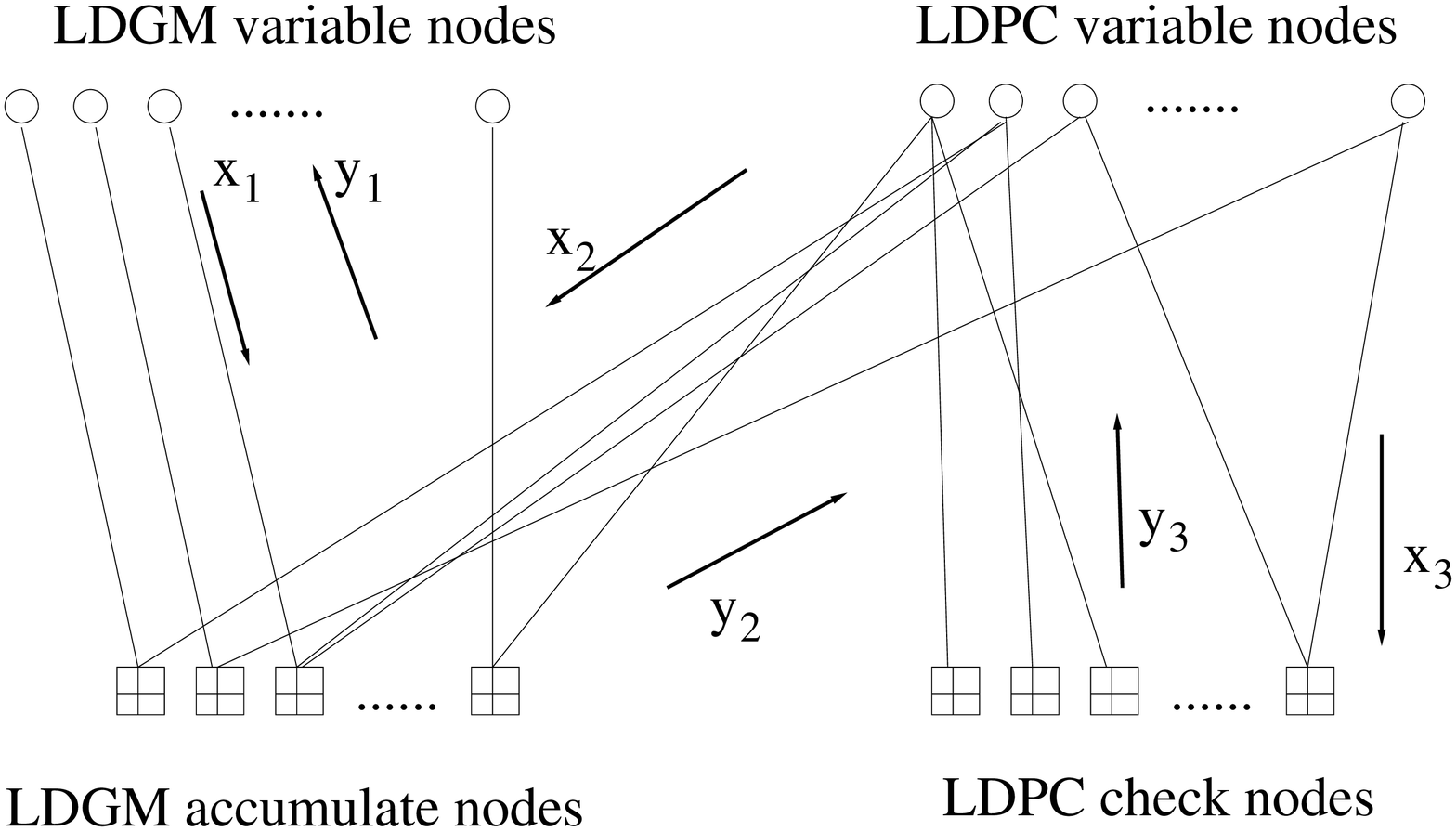}
\caption{The erasure probabilities for the LDGM-LDPC codes}\label{factor-g-pic}
\end{figure}

\begin{align}
\label{de1}
        x_1^l = \delta y_1^{l-1}
\end{align}
Consider the message on a random edge from an LDGM accumulate node to a LDGM variable node. An erasure results when all at least one message from the LDPC variable nodes in the previous iteration were erasures. 
\begin{align}
\label{de2}
y_1^l = 1 - R_{\text{G}}(1- x_2^{l-1})
\end{align}
where $R_{\text{G}}(x) = \frac{\int_0^x \rho_{\text{G}}(t) dt}{\int_0^1 \rho_{\text{G}}(t) dt}$. 
Consider a random edge from an LDPC variable node to the LDGM accumulate node. An erasure results on this edge during the $l$th iteration if all the incoming edges are erasures and the variable node was erased or punctured. Thus,
\begin{align}
\label{de3}
x_2^l = (1 - (1-\delta)(1-p)) \lambda_G(y_2^{l-1})L_H(y_3^{l-1})
\end{align}
where $p$ is the puncturing probability and $L_{\text{H}}(x) = \frac{\int_0^x \lambda_{\text{H}}(t) dt}{\int_0^1 \lambda_{\text{H}}(t) dt}$. The probability along a random edge from a LDGM accumulate node to a LDPC variable node in the $l$th iteration happen if any of the channel outputs are erased in the previous iteration.
\begin{align}
\label{de4}
        y_2^l = 1 - (1-x_1^{l-1})\rho_G(1-x_2^{l-1})
\end{align} 
Along the lines of \eqref{de3}, a randomly chosen edge from an LDPC variable node to an LDPC check-node has an erasure in the $l$th iteration if the variable node experienced an erasure and all incoming edges carried erasure messages in the $l-1$th iteration.
\begin{align}
\label{de5}
        x_3^l = (1 - (1-\delta)(1-p))L_G(y_2^{l-1})\lambda(1-y_3^{l-1})
\end{align}
An erasure along a randomly chosen edge from an LDPC check node to an LDPC variable node happens if any incoming edge has an erasure.
\begin{align}
\label{de6} 
y_3^l = 1 - \lambda_H(1-x_3^{l-1})
\end{align}

\begin{mydef}
The \emph{fixed-points} of density-evolution described in (\ref{de1}-\ref{de6}) are defined as 
\begin{align*}
\text{lim}_{l \rightarrow \infty} x_i^l &= x_i \text{ and } \text{lim}_{l \rightarrow \infty} y_i^l = y_i,& \forall i \in \{1,2,3\}
\end{align*}
\end{mydef}
We solve for $x_2$ and $x_3$ from (\ref{de1}-\ref{de6}) and obtain,
\begin{equation}
\begin{split}
\label{x2x3eq}
x_2 &= [1 - (1-\delta)(1-p)] . \\ &\lambda_G(1 - (1-\delta(1 - R_G(1-x_2)))\rho_G(1-x_2)). \\ &  L_H(1 - \rho_H(1-x_3)) \\
x_3 &= [1 - (1-\delta)(1-p)] . \\ & L_G(1 - (1-\delta(1 - R_G(1-x_2)))\rho_G(1-x_2)). \\ &  \lambda_H(1 - \rho_H(1-x_3)) 
\end{split}
\end{equation}

\begin{thm}
\label{thm3}
Consider the LDGM-LDPC code ensemble as defined in lemma \ref{ldgmlem}. The point $x_2 = 0$ is stable during density-evolution if
\begin{align}
\label{finaleq}
(1 - (1-\delta)(1-p))^2 \lambda_G(0) L_G^{'}(0)\rho_H^{'}(1) \lambda_H(0) L_H^{'}(0) \\ . [ \delta L_G^{'}(1) + \rho_{G}^{'}(1) ] < 1
\end{align}
\end{thm}
The theorem is proved in the appendix.\par As per lemma \ref{pepsilonlem}, at rates very close to capacity, if the puncturing rate $p = 1 - \kappa\epsilon$, the lower-bound on the complexity is finite as the rates are arbitrarily close to capacity. We study the stability condition when the rates are chosen very close to capacity.
\begin{lem}
When the code rate of the LDGM-LDPC code is arbitrarily close to capacity i.e. $\epsilon \rightarrow 0$, and $p = 1 - \kappa\epsilon$, the stability condition for iterative decoding is
\begin{align*}
\lambda_G(0) L_G^{'}(0)\rho_H^{'}(1) \lambda_H(0) L_H^{'}(0) [ \delta L_G^{'}(1) + \rho_{G}^{'}(1) ] < 1
\end{align*}
\end{lem}
\begin{proof}
Substituting $p = 1 - \kappa\epsilon$ and $\epsilon \rightarrow 0$ in \eqref{finaleq} proves the above lemma.
\end{proof}

\section{Conclusion}
\label{conclusion}
Irregular LDGM-LDPC codes have been shown to be LDPC code with some randomly punctured bits. The ensemble of  irregular LDGM-LDPC codes have been shown to achieve the capacity of the BISOM channel with bounded complexity. The stability condition for the punctured LDGM-LDPC codes over the BEC under message-passing decoding was obtained.  

\section{Discussion}
This paper obtains lower-bounds on the complexity of the decoding-complexity of irregular LDGM-LDPC codes. These bounds are existential in nature and indicate the existence of LDGM-LDPC codes that achieve rates arbitrarily close to capacity with puncturing. 

\section{Acknowledgments}
The authors would like to thank Yury Polyanskiy for helpful comments and discussions.

\appendix
\subsection{Proof of Theorem \ref{ratebound}}
\begin{proof}
From lemma \ref{pblem}, proposition \ref{propsason}, setting $n \rightarrow \infty$ and $P_b \rightarrow 0$,  
\begin{equation}
\begin{split}
\label{entrzero}
0 \ge 1 - \sum_{j = 1}^J p_j C_j - (1 - R_d) 
 \Biggl(1 - \frac{1}{2n(1 - R_d)\log 2} \\
.\sum_{p = 1}^{\infty} \left\{ \frac{1}{p(2p -1)} \sum_{m = 1}^{n(1 - R_d)} 
\prod_{j = 1}^J (g_{j,p})^{\beta_{j,m}} \right\} \Biggr) 
\end{split}
\end{equation}
By considering the first term of the sum in $p$ in the above equation, the above equation can be bounded as follows:
\begin{equation}
\begin{split}
\label{entrzero2}
0 & \ge 1 - \sum_{j = 1}^J p_j C_j - (1 - R_d) 
 \Biggl(1 - \frac{1}{2n(1 - R_d)\log 2} \\&
.\left\{  \sum_{m = 1}^{n(1 - R_d)} 
\prod_{j = 1}^J (g_{j,1})^{\beta_{j,m}} \right\} \Biggr) 
\end{split}
\end{equation}

Since the first $n^{[1]}$ bits pass through the BISOM channel without puncturing,  $C_1 = C$. Further, since the last $n^{[2]}$ bits of the codeword are punctured, $C_2 = (1-p)C$. Let $c$ be the number of parity checks in the matrix (\ref{hmatrix}). Then, $c = n(1-R_d)$. Due to the structure of the code,  from (\ref{betas}) and (\ref{hmatrix}), for $m \in [1,n^{[1]}]$, $\beta_{1,m} = 1$ and $\beta_{2,m}$ is distributed as $\rho_G(x)$, (defined in (\ref{lambdag})). Further, for $m \in [n^{[1]}+1,c]$, $\beta_{1,m} = 0$ and $\beta_{2,m}$ is distributed as $\rho_H(x)$ (defined in (\ref{lambdah})). We compute the expectation of (\ref{entrzero2}) over the distributions $\rho_G(.)$ and $\rho_H(.)$. The expectation ${\bf E} \Bigl[\sum_{m = 1}^{n(1 - R_d)} 
\prod_{j = 1}^J (g_{j,p})^{\beta_{j,m}}\Bigr]$ is computed as follows:
\begin{equation}
\begin{split}
\label{exp1}
{\bf E}  \sum_{m = 1}^{n(1 - R_d)} \prod_{j = 1}^J (g_{j,p})^{\beta_{j,m}}  &= {\bf E} \sum_{m = 1}^{n^{[1]}} \prod_{j = 1}^J (g_{j,p})^{\beta_{j,m}} + \\&  {\bf E} \sum_{m = n^{[1]}+1}^{c} \prod_{j = 1}^J (g_{j,p})^{\beta_{j,m}} 
\end{split}
\end{equation}
${\bf E} \sum_{m = 1}^{n^{[1]}} \prod_{j = 1}^J (g_{j,p})^{\beta_{j,m}}$ is evaluated as follows:
\begin{equation}
\label{exp2}
\begin{split}
{\bf E} &\sum_{m = 1}^{n^{[1]}} \prod_{j = 1}^J (g_{j,p})^{\beta_{j,m}} \overset{a}{=} n^{[1]} {\bf E}_{\beta_{2,m}} [g_{1,p} g_{2,p}^{\beta_{2,m}}] \\& \overset{b}{=} n^{[1]}  g_{1,p} {\bf E}_{\beta_{2,m}} [g_{2,p}^{\beta_{2,m}}] \overset{c}{\ge} n^{[1]}  g_{1,p} g_{2,p}^{{\bf E}_{\beta_{2,m}} \beta_{2,m}} 
\end{split}
\end{equation}
where $\overset{a}{=}$ follows because $\beta_{1,m}=1$, $\overset{b}{=}$ follows because $g_{1,p}$ is a constant as the expectation is w.r.t. $\beta_{2,m}$, $\overset{c}{\ge}$ follows from the convexity of the function $g_{2,p}^{\beta_{2,m}}$ and the Jensen's inequality. 
${\bf E} \sum_{m = n^{[1]}+1}^{c} \prod_{j = 1}^J (g_{j,p})^{\beta_{j,m}}$ is evaluated as follows:
\begin{equation}
\begin{split}
\label{exp3}
{\bf E} &\sum_{m = n^{[1]}+1}^{c} \prod_{j = 1}^J (g_{j,p})^{\beta_{j,m}} \overset{d}{=}   c_H   {\bf E}_{\beta_{2,m}} [g_{2,p}^{\beta_{2,m}}]  \overset{e}{\ge} c_H   g_{2,p}^{{\bf E}_{\beta_{2,m}} \beta_{2,m}} 
\end{split}
\end{equation}
where $c_H = c - n^{[1]}$ is the number of parity-checks in the lower LDPC layer of the LDGM-LDPC code, where $\overset{d}{=}$ follows because $\beta_{1,m}=0$, $\overset{e}{\ge}$ follows from the convexity of the function $g_{2,p}^{\beta_{2,m}}$ and the Jensen's inequality. From (\ref{exp2}) and (\ref{exp3}), the sum in (\ref{exp1}) becomes
\begin{equation}
\begin{split}
\label{exp4}
{\bf E}  &\sum_{m = 1}^{n(1 - R_d)} \prod_{j = 1}^J (g_{j,p})^{\beta_{j,m}} \ge n^{[1]}  g_{1,p} g_{2,p}^{{\bf E}_{\beta_{2,m}} \beta_{2,m}} + c_H   g_{2,p}^{{\bf E}_{\beta_{2,m}} \beta_{2,m}} \\& \overset{f}{=} n^{[1]}  g_{1,p} g_{2,p}^{a_L} + c_H   g_{2,p}^{a_R}
\end{split}
\end{equation}
where $\overset{f}{=}$ results by replacing the average number of edges to the LDGM accumulate nodes and LDPC check nodes by $a_L$ and $a_R$ respectively.
The right hand side of (\ref{exp4}) is further simplified as follows.
\begin{equation}
\begin{split}
\label{exp5}
n^{[1]}  g_{1,p} g_{2,p}^{a_L} &+ c_H   g_{2,p}^{a_R} =  (n^{[1]}  g_{1,p} + c_H)\Bigl[\frac{n^{[1]}  g_{1,p} }{n^{[1]}  g_{1,p} + c_H}g_{2,p}^{a_L} \\&+ \frac{c_H}{n^{[1]}  g_{1,p} + c_H}g_{2,p}^{a_R}\Bigr] \\& \overset{g}{\ge} (n^{[1]}  g_{1,p} + c_H)g_{2,p}^{\frac{n^{[1]}  g_{1,p} }{n^{[1]}  g_{1,p} + c_H}a_L + \frac{c_H}{n^{[1]}  g_{1,p} + c_H}a_R}
\end{split}
\end{equation}
$\overset{f}{\ge}$ is explained as follows. Consider a random-variable $B$ with a probability distribution defined as:
\begin{equation*}
P_B(b) = \begin{cases} \frac{n^{[1]}  g_{1,p} }{n^{[1]}  g_{1,p} + c_H},&  b = a_L \\ \frac{c_H}{n^{[1]}  g_{1,p} + c_H},& b = a_R \end{cases} 
\end{equation*}
Consider $f(B) = g_{2,p}^B$. As $f(B)$ is convex in $B$, from the Jensen's inequality, $\overset{f}{\ge}$ follows.
From (\ref{entrzero2}-\ref{exp5}), and substituting $n(1-R_d) = c = n^{[1]}+c_H$,
\begin{equation}
\begin{split}
\label{entrzero3}
0 & \ge 1 - \sum_{j = 1}^J p_j C_j - (1 - R_d) 
 \Biggl(1 - \frac{1}{2\log 2} \\&
.\left\{  \frac{n^{[1]}  g_{1,1} + c_H}{n^{[1]}+c_H}g_{2,1}^{\frac{n^{[1]}  g_{1,1} }{n^{[1]}  g_{1,1} + c_H}a_L + \frac{c_H}{n^{[1]}  g_{1,1} + c_H}a_R} \right\} \Biggr) 
\end{split}
\end{equation}
We make the following substitutions in the above equations $c_H = (1-R_H)n^{[2]}$, $n^{[1]} = p_1 n$ and $n^{[2]} = p_2 n$, where $R_H$ is the rate of the lower LDPC code: 
\begin{equation}
\begin{split}
\label{entrzero3}
0 & \ge 1 - \sum_{j = 1}^J p_j C_j - (1 - R_d) 
 \Biggl(1 - \frac{1}{2\log 2} \\&
.\left\{  \frac{p_1  g_{1,1} + (1-R_H)p_2}{p_1+(1-R_H)p_2}g_{2,1}^{\frac{p_1  g_{1,1} a_L + (1-R_H)p_2 a_R}{p_1  g_{1,1}+ (1-R_H)p_2}} \right\} \Biggr) 
\end{split}
\end{equation}
By replacing $C_1 = C$, $C_2 = (1-p)C$ and solving for $R_d$ in the above equation, we obtain the desired bound.
\end{proof}

\subsection{Proof of Theorem \ref{thm3}}
\begin{proof}
The equations \eqref{x2x3eq} can be represented as 
\begin{align*}
x_2 &= \psi_A(x_2,x_3) ,& x_3 &= \psi_B(x_2,x_3)
\end{align*}
Consider a fixed point in the density-evolution $(x_2,x_3) = (x_2^o,x_3^o)$. The above functions can be linearly approximated in the neighbor of the fixed point as follows,
\begin{align}
\label{x2A}
\psi_A(x_2,x_3) = x_2^o + \Biggl[ \frac{\partial \psi_A}{\partial x_2}  + \frac{\partial \psi_A}{\partial x_3} \frac{d x_3}{d x_2}\Biggr](x_2 - x_2^o) + o(x_2 - x_2^o)^2
\end{align} 
Since $x_3 = \psi_B(x_2,x_3)$, taking derivatives on both sides,
\begin{align}
\label{x2B}
\frac{d x_3}{d x_2} = \frac{\partial \psi_B}{\partial x_2} + \frac{\partial \psi_B}{\partial x_3}\frac{d x_3}{d x_2}
\end{align}
Substituting $\frac{d x_3}{d x_2}$ from \eqref{x2B} into \eqref{x2A},
\begin{align*}
\psi_A(x_2,x_3) &= x_2^o + \Biggl[ \frac{\partial \psi_A}{\partial x_2}  + \frac{\partial \psi_A}{\partial x_3} \frac{\frac{\partial \psi_B}{\partial x_2}}{1 - \frac{\partial \psi_B}{\partial x_3}} \Biggr](x_2 - x_2^o) \\+ o(x_2 - x_2^o)^2
\end{align*}
For stability, $\Biggl[ \frac{\partial \psi_A}{\partial x_2}  + \frac{\partial \psi_A}{\partial x_3} \frac{\frac{\partial \psi_B}{\partial x_2}}{1 - \frac{\partial \psi_B}{\partial x_3}} \Biggr] < 1$. Evaluating the derivatives and substituting $x_2^o = x_3^o = 0$, the result is obtained.
\end{proof}

\end{document}